\documentclass[12pt]{article}
\usepackage{fullpage,amsmath,amssymb,amsfonts,amsthm}

\usepackage{a4wide}
\usepackage{url}
\usepackage{hyperref}
\usepackage{algorithm,algorithmic}
\usepackage{xcolor}
\usepackage{graphicx}
\usepackage[margin=1cm, font=footnotesize]{caption}
\usepackage{authblk}
\usepackage{soul}

\usepackage{mathtools}
\usepackage{hyperref}

\newtheorem{theo}{Theorem}[section]
\newtheorem{defi}[theo]{Definition}
\newtheorem{prop}[theo]{Proposition}
\newtheorem{coro}[theo]{Corollary}
\newtheorem{lemma}[theo]{Lemma}

\DeclareMathOperator{\rank}{rank}
\DeclareMathOperator{\Ker}{Ker}

\DeclareMathOperator{\atanh}{atanh}

\newcommand{\R}{\mathbb{R}}

\newcommand{\Z}{\mathbb{Z}}
\newcommand{\N}{\mathbb{N}}

\DeclareMathOperator{\Int}{Int}

\DeclareMathOperator{\cov}{cov}
\newcommand{\rcov}{{\rho_{\cov}}}
\newcommand{\rcovbar}{{\bar{\rho}_{\cov}}}

\newcommand{\error}{x}
\newcommand{\syndrome}{\sigma}

\newcommand{\vq}{V_Q}
\newcommand{\vc}{V_C}
\newcommand{\cE}{{\cal E}}

\DeclareMathOperator{\inj}{inj}
\DeclareMathOperator{\Vol}{Vol}
\newcommand{\Rinj}{{R_{\inj}}}
\newcommand{\cM}{{\cal M}}
\newcommand{\cMi}{{{\cal M}_i}}

\newcommand{\HypD}{\mathbb{H}^D}

\newcommand{\TD}{{\mathbb T}_L^D} 
\DeclareMathOperator{\TC}{TC}
\newcommand{\cA}{{\cal A}}

\title{Toward a Union-Find decoder for quantum LDPC codes}

\author[1]{Nicolas Delfosse}
\author[2]{Vivien Londe}
\author[1]{Michael Beverland}

\affil[1]{Microsoft Quantum and Microsoft Research, Redmond, USA}
\affil[2]{Microsoft France, Issy-les-Moulineaux, Paris}

\begin{document}

\maketitle

\begin{abstract}
Quantum LDPC codes are a promising direction for low overhead quantum computing. 
In this paper, we propose a generalization of the Union-Find decoder as a decoder for quantum LDPC codes. 
We prove that this decoder corrects all errors with weight up to $An^\alpha$ for some $A, \alpha > 0$ for different classes of quantum LDPC codes such as toric codes and hyperbolic codes in any dimension $D \geq 3$ and
quantum expander codes.
To prove this result, we introduce a notion of covering radius which measures the spread of an error from its syndrome.
We believe this notion could find application beyond the decoding problem.
We also perform numerical simulations, which show that our Union-Find decoder outperforms the belief propagation decoder in the low error rate regime in the case of a quantum LDPC code with length 3600.
\end{abstract}

\section{Introduction}

Classical Low Density Parity Check (LDPC) codes \cite{gallager1962LDPC} are ubiquitous in modern communication and information processing. 
They achieve a high rate of information storage, a large minimum distance and are equipped with a linear-time decoder whose performance is close to that of the maximum likelihood decoder: the Belief Propagation (BP) decoder.
Quantum LDPC codes could bring similar advantages to quantum computing.
Moreover, quantum LDPC codes are defined by low weight checks which makes them easier to implement in practice than more general quantum codes. This is because measurements of high-weight checks are likely too noisy to be useful with quantum hardware. 
It is possible to apply the BP decoder in the quantum setting,  although unfortunately it does not tend to perform well for quantum LDPC codes \cite{poulin2008iterative}.

\medskip
In this work, we propose a generalization of the Union-Find decoder \cite{delfosse2017UF_decoder} to decode quantum LDPC codes. 
We analyze its performance for different classes of quantum LDPC codes including toric codes, hyperbolic codes, and expander codes.
toric codes, with local checks on a $D$-dimensional grid, seem easier to implement, while hyperbolic and expander codes can offer a higher encoding rate.

\medskip
The generalization of the BP decoder to the quantum setting was first considered by Mackay, Mitchison and McFadden in 2004 \cite{mackay2004sparse}.
The authors emphasize the advantage of high-rate quantum LDPC codes over other quantum error correction schemes such as quantum Reed-Muller codes \cite{steane1999quantum} or quantum Bose–Chaudhuri–Hocquenghem (BCH) codes \cite{grassl1999QBCH}.
They also note that the Tanner graph contains an enormous number of short cycles which degrade the performance of the BP decoder.
Poulin and Chung explored different modifications of this decoder and observed that the decoding remains the main point of failure in LDPC quantum error correction schemes \cite{poulin2008iterative}.
To improve the performance of the BP decoder, one could also consider moving away from standard quantum LDPC codes with asymmetric codes obtained by combining BCH codes and LDPC codes \cite{sarvepalli2009asymmetric}.

\medskip
In the classical setting, Sipser and Spielman observed that the expansion property of the Tanner graph of LDPC codes can be exploited to design an efficient decoder for expander codes \cite{sipser1996expander}.
This algorithm was generalized to the quantum setting by Leverrier, Tillich and Zémor \cite{leverrier2015quantum_expander_codes}, providing an alternative to the BP decoder for quantum LDPC codes in the special case of quantum expander codes.
The basic idea of this decoding algorithm is to look for an error in a constant size neighborhood of a qubit such that correcting this error would reduce the weight of the syndrome.
This local correction is repeated until the syndrome becomes trivial.
The authors proved that such a local error always exists if the Tanner graph has a sufficiently large expansion.
However, if the expansion is not large enough, the decoder might get stuck with a non-trivial syndrome that cannot be reduced locally.
Hasting's decoder for 4D hyperbolic codes, which removes errors in a hyperbolic space by locally contracting them, has a similar flavor \cite{hastings2013decoding}.
These strategies rely on Tanner graphs with large expansion are are therefore unlikely to perform well with toric codes or with quantum LDPC codes with moderate expansion.

\medskip
The decoding of 2D surface codes and toric codes has been extensively studied and several decoders achieve a satisfying performance such as the Minimum Weight Perfect Matching decoder \cite{dennis2002topological} or the Union-Find decoder \cite{delfosse2017UF_decoder}.
Different decoders have been considered for toric codes in dimension $D \geq 3$ such as the Renormalization Group (RG) decoder \cite{bravyi2011analytic} or decoders based on local rules \cite{breuckmann2016local, kubica2019cellular}.
To emphasize the difficulty to generalize the RG decoder to quantum LDPC codes, let us review the basic principle of this decoder in the case of a 2D toric code over an $L \times L$ square lattice.
During the first step of the RG decoder, the lattice is decomposed into disjoint patches with size $\ell \times \ell$ and the error inside each patch is estimated locally.
Then, incident patches are merged together to form $2\ell \times 2\ell$ patches and the error inside a merged patch is derived from the results of the first step.
The RG decoder identifies the error by repeating this procedure $O(\log L)$ times until all the patches are merged.
This strategy applies to all local codes defined on cubic lattices but it is unclear to us how to build these self-similar patches for hyperbolic codes or for general quantum LDPC codes.
The toric code decoder of \cite{kubica2019cellular} has the advantage of being based on a local rule. 
Like the quantum expander code decoder, it determines a local correction but the local correction is not chosen to minimize the syndrome weight. Instead, it moves the syndrome in a fixed direction of the cubic lattice supporting the code. 
We do not know how to adapt this notion of direction, which seems essential, to general quantum LDPC codes.

\medskip
In this article, we describe a generalization of the Union-Find decoder \cite{delfosse2017UF_decoder} to quantum LDPC codes.
For simplicity, we consider the setting of perfect syndrome extraction. 
The Union-Find decoder was originally proposed for two-dimensional topological codes.
Our variant of the Union-Find decoder is very general and it can be used to decode any stabilizer code.
We prove that it achieves good performance for locally testable quantum codes, including quantum expander codes \cite{leverrier2015quantum_expander_codes}, $D$-dimensional hyperbolic topological codes such as the 4D hyperbolic codes of Guth and Lubotsky \cite{guth20144D_codes}, and also $D$-dimensional toric codes in any dimension $D \geq 3$.
Precisely, for these families of codes, we show that the Union-Find decoder is capable of correcting all errors with weight up to $A n^\alpha$ for some positive constants $A$ and $\alpha$, where $n$ is the length of the code.
Unlike the quantum expander code decoder, our approach is not limited to codes with high expansion.
It is satisfying to have a unified decoding strategy for quantum LDPC codes. 
We believe that this work makes an important step toward a flexible, low-complexity, fault-tolerant decoder that performs well with most quantum LDPC codes.

\medskip
To analyze the performance of the decoder, we introduce the notion of the covering radius of an error, which measures the spread of the error from its syndrome.
We show that our generalization of the Union-Find decoder successfully corrects all errors with small covering radius. 
Then, by establishing bounds on the covering radius of errors for different classes of quantum LDPC codes, we prove that the Union-Find decoder succeeds.
We believe that this notion of the covering radius could find applications beyond the decoding problem.

\medskip
In Section~\ref{sec:background}, we review the three classes of quantum LDPC codes considered in this article.
Our version of the Union-Find decoder is described in Section~\ref{algo:uf_decoder} and we demonstrate in Section~\ref{sec:uf_success} that our decoder successfully corrects errors with small covering radius.
Section~\ref{sec:hyp_codes} and Section~\ref{sec:loc_testable_codes} provide a bound on the covering radius of errors in hyperbolic codes and locally testable codes, showing that the decoder performs well for these codes.
Toric codes are considered in Section~\ref{sec:toric_codes}. To prove that the Union-Find decoder performs well for toric codes we use a bound on the covering radius in combination with the disjointness property of logical errors \cite{jochym2018disjointness, webster2020universal}.
In Section~\ref{sec:numerics}, we perform numerical simulations which show that our Union-Find decoder outperforms the BP decoder in the low error rate regime in the case of a quantum LDPC code with length 3600.

\section{Background and Notation} \label{sec:background}

\subsection{Graphs}

A {\em graph} is a pair of sets $G = (V, E)$ where $V$ is a set called the vertex set and $E$ is a set of pairs of elements of $V$ called the edge set.
We assume that edges are not oriented, that is $\{u, v\} = \{v, u\} \in E$.
The {\em degree} of a vertex $u \in V$ in $G$ is the number of edges incident to $u$, that is the number of edges containing $u$.
The degree of the graph $G$ is the maximum degree of a vertex of $G$.

\medskip
A graph is equipped with a natural notion of distance induced by the edge set. 
The distance $d(u, v)$ between two vertices $u$ and $v$ in $G$ is the minimum number of edges of a path connecting them.
If no such path exists, we define $d(u, v) = +\infty$.
The ball of radius $R$ centered at $u \in V$, denoted $B_G(u, R)$ is the set of vertices $v$ of $G$ such that $d(u, v) \leq R$.
The set of vertices at distance exactly $R$ from $u$, denoted $S_G(u, R)$ is called the sphere of radius $R$ centered at $u \in V$.
More generally, for a subset of vertices $A \subseteq V$, the set
$
B_G(A, R)
$
is the set of nodes at distance up to $R$ from any vertex of $A$.

\medskip
The {\em connected component} of a vertex $u$ in $G$ is the set of vertices (also called nodes) that can be reached from $u$ by a path of edges of $G$.
The connected components of $G$ form a partition of the vertex set.
We will also consider the connected components of a subset $A \subseteq V$. In this case, the connected component of $u \in A$ is the set of vertices of $A$ that can be reached from $u$ by a path included in $A$.

\medskip
Given $A \subseteq V$, the {\em interior} of the set $A$, denoted $\Int(A)$, is defined to be the set of vertices $u$ of $A$ such that all the neighbors of $u$ are also included in $A$, {\em i.e.} $u \in \Int(A)$ iff $B_G(u, 1) \subseteq A$.

\subsection{Manifolds}
\label{sec:manifolds}

In this section, we review the basic properties of Euclidean and hyperbolic manifolds, upon which Kitaev's topological LDPC codes are constructed \cite{kitaev2003top_codes}. 
For simplicity, we focus on closed manifolds. 
We choose the quotient construction of closed manifolds because it leads to more intuitive definitions of the metric and the injectivity radius which play a crucial role in this article. We will then consider well-behaved cellulations of these manifolds which can be used to build topological codes.

\medskip 
Let us first review the construction of a $D$-dimensional torus as a quotient space. 
The $D$-dimensional Euclidean space is denoted $\R^D$. 
It is equipped with the standard Euclidean metric. 
Let us consider a discrete subgroup of the isometries of $\R^D$, namely the group of translations $T$ of the form $x \mapsto x + t$ where $t \in \Z^D$. 
The quotient space $\R^D / T$, defined by identifying all the points $\tau(x)$ with $\tau \in T$ with the point $x$, is a $D$ dimensional torus. 

\medskip 
Any closed Euclidean (respectively hyperbolic) $D$-manifold can be obtained as a quotient $\cM = U / G$ where $U = \R^D$ (respectively $U = \HypD$) and where $G$ is a discrete subgroup of the isometries of $U$ \cite{hatcher2005algebraic_top, jones1987complex_functions}.
The space $\HypD$ is the upper half-space $\{ x = (x_1, \dots, x_D) \in \R^D \ | \ x_D > 0 \}$ equipped with a hyperbolic metric. The definition of the hyperbolic metric and basic properties of $\HypD$ can be found in \cite{jones1987complex_functions}.
The quotient map 
$
\pi: U \rightarrow \cM = U / G
$
is a surjective continuous map. Moreover, it induces a local homeomorphism between the space $U$ and the manifold $\cM$. This is because for any point $p$ in $U$, for a sufficiently small radius $R$, the restriction of the map $\pi$ to the the ball $B_U(p, R)$ in $U$ is a bijection\footnote{By $B_U(p, R)$ we mean the set of points of $U$ at distance less or equal than $R$ from $p$.}.
The injectivity radius measures the maximum radius of this local isomorphism. Precisely, the {\em injectivity radius} of $\cM$ is defined to be
\begin{align} \label{eq:Rinj}
\Rinj = \sup \{ R \in \R \ | \ \pi_{|B_U(p, R)} \text{ is injective for all } p \in U \},
\end{align}
where $\pi_{|B_U(p, R)}$ is the restriction of $\pi$ to the ball with radius $R$ centered in $p$ in the space $U = \R^D$ or $\HypD$.
The manifold $\cM$ is equipped with the metric induced by the metric of $U$. This allows us to compute the volume of an $i$-dimensional region $c$ of the manifold $\cM$ that we denote $\Vol_i(c)$.
In the case of the torus considered above, the quotient map is $\pi: (x_1,x_2, \dots x_D) \rightarrow (x_1 \mod 1,x_2 \mod 1, \dots x_D \mod 1)$, the injectivity radius is $0.5$, and volumes are calculated using the standard Euclidean metric.

\medskip
Kitaev's topological codes can be defined from a {\em cellulation} of a closed manifold $\cM = U / G$, that is a partition of the manifolds into polytopes (see \cite{hatcher2005algebraic_top} for a formal definition). Let $\cM_i$ be the set of $i$-cells of the cellulation for $i=0, \dots, D$.
In order to avoid pathological cellulations, we assume that all $i$-cells of $\cM$ have roughly the same volume, {\em i.e.} for all $i=1, \dots, D$ there exist constants $v_i, v_i' > 0$ such that
\begin{align} \label{eq:cells_vol}
v_i \leq \Vol_i(c) \leq v_i',
\end{align}
for every $i$-cell $c \in \cMi$. This condition also guarantees that the cellulation is finite because the volume of the manifold is finite.
All cellulations considered in this article satisfy Eq.~\eqref{eq:cells_vol}.

\medskip
The space of $i$-{\em chains}, denoted $C_i$, is the $\Z_2$-linear space whose basis is indexed by the $i$-cells of $\cM$. A $i$-chain is a sum $\sum_{c \in \cMi} \lambda_c c$ with $\lambda_c \in \Z_2$. Any $i$-chain can be interpreted as a subset of $i$-cells (given by its support) and {\em vice versa}.

\medskip
For all $i=1, \dots, D$, the {\em boundary map} is a $\Z_2$-linear map from $C_i$ to $C_{i-1}$.
By linearity, it is sufficient to define $\partial(c)$. 
The image $\partial(c)$ of $c \in C_i$ is the $(i-1)$-chain whose support is the set of $(i-1)$-cells included in $c$.

\medskip
In the case of the $D$-dimensional torus $\R^D / \Z^D$ considered previously, we can build a cellulation with $L^D$ cubic cells of length $1/L$. 
The 0-cells  are vertices, 1-cells are edges, 2-cells are squares, 3-cells are cubes, {\em etc}. Naturally, in this example, the boundary of a 3-cell contains six 2-cells and the boundary of a 2-cell contains four 1-cells.


\subsection{Quantum error correction}
\label{sec:qec}

The Union-Find decoding algorithm we introduce in Sec.~\ref{sec:uf_success} applies to a broad class of quantum error correcting codes known as Calderbank Shor Steane (CSS) codes~\cite{calderbank_good_1996,steane_simple_1996}.
A {\em CSS code} is fully specified by a pair of binary {\em parity-check matrices} $H_X \in M_{r_X \times n}(\Z_2)$ and $H_Z \in M_{r_Z \times n}(\Z_2)$ such that $H_X H_Z^T = 0$.
This pair of matrices defines a quantum code which encodes $k$ logical qubits into $n$ qubits where $k = n - \rank(H_X) - \rank(H_Z)$.

\medskip
An error for a CSS code with length $n$ can be represented as a pair of binary vectors $e = (e_X, e_Z) \in \Z_2^n \times \Z_2^n$.
The two components of the error can be corrected independently with the same procedure. In what follows, we focus on the correction of the $X$ part $e_X$ of the error.
When quantum data is encoded using a CSS code with parity-check matrices $H_X$ and $H_Z$, there is a set of errors, called {\em stabilizers}, that have no effect on the encoded data. An $X$ error $e_X$ is a stabilizer iff it belongs to the row space of $H_X$.

\medskip
The {\em syndrome} of the error $e_X$ is the vector $\sigma_X = H_Z e_X^T \in \Z_2^{r_Z}$. Each row of the matrix $H_Z$ defines a $Z$ check providing one syndrome bit.
In this work, we assume that the syndrome can be measured perfectly. The goal of the decoder is to identify the error $e_X$ given its syndrome. 
The decoder succeeds if the estimation $\hat e_X$ it returns is equivalent to $e_X$ up to a stabilizer, {\em i.e.} if $e_X + \hat e_X$ is a stabilizer.

\medskip
Based on the equivalence of errors up to stabilizers, the $X$ {\em minimum distance} of a CSS code, denoted $d_X$, is defined to be the minimum weight of an $X$ error $e_X$ with trivial syndrome which is not a stabilizer. The $Z$ minimum distance $d_Z$ is defined similarly by swapping the roles of $X$ and $Z$.
The minimum distance of a CSS code is defined as $d = \min\{d_X, d_Z\}$.
The parameters of a CSS code are denoted $[[n, k, d]]$.

\subsection{Quantum LDPC codes}
\label{sec:LDPC-quantum-codes}

A CSS stabilizer code is said to be a Low Density Parity-Check (LDPC) quantum code if both the parity-check matrices $H_X$ and $H_Z$ are sparse.
LDPC quantum codes are promising for practical purposes because their stabilizer measurements are of bounded weight.
In the remainder of this section we define some important families of LDPC quantum codes for which we will later prove bounds and numerically simulate with the Union-Find decoder.

\subsubsection{Topological codes}

Topological codes are among the most well-known LDPC quantum codes, and are defined in terms of qubits placed on manifolds as described in Sec.~\ref{sec:manifolds}.

\medskip
First, let us review Kitaev's construction of {\em topological codes} {\cite{kitaev2003top_codes}}. Let $\cM$ be a closed $D$-manifold and let $\cMi$ be the set of $i$-cells of a finite cellulation of the manifold for $i=0, \dots, D$. 
Kitaev's construction associates a CSS code with the $m$-cells of the cellulation for some $1 \leq m \leq D-1$. A qubit is placed on each $m$-cell. For each $(m+1)$-cell $c$, we define an $X$ stabilizer whose support is given by the qubits on the boundary of $c$. Similarly, a $Z$ check is associated with each $(m-1)$-cell $c$ and its support is the set of $m$-cells $c'$ such that $c$ belongs to the boundary on $c'$.
We will consider two classes: the \emph{$D$-dimensional toric codes} built from cubic cellulations of the $D$-dimensional torus, and \emph{hyperbolic codes} built from finite cellulations of hyperbolic manifolds.

\medskip
The major advantage of toric codes is that they can be naturally embedded in a cubic lattice and they achieve a large minimum distance. 
However, they come with a poor encoding rate. The parameters of $D$-dimensional toric codes are constrained by the tradeoff $k d^{2/(D-1)} \leq c n$ for some constant $c$ \cite{bravyi2010tradeoffs}.

\medskip
Hyperbolic codes in dimension $D>2$ are especially interesting because they can achieve a non-vanishing rate $k/n$ and a large minimum distance $d = n^a$ for some $a>0$ \cite{guth20144D_codes}. These parameters are not achievable with two-dimensional hyperbolic codes \cite{delfosse2013tradeoffs} nor with toric codes in any dimension.

\subsubsection{Locally-testable codes}

Locally-testable quantum codes, introduced in \cite{aharonov2015QLT}, are of interest due to their resilience against faulty measurements which can lead to single-shot error correcting properties \cite{campbell2019single_shot}.
In this section we recall the definition of quantum locally testable codes through the soundness property formalized in \cite{eldar2017local}. 
The local testability of different families of quantum codes is studied in \cite{hastings2016quantum} and \cite{leverrier2019towards}.

\medskip
An $X$ error $e_X$ with syndrome $\syndrome_X$ is said to be a {\em reduced error} if it is a minimum weight error with syndrome $\syndrome_X$. 
The soundness property that we introduce now is the defining property of locally testable CSS codes.
A CSS code is said to have {\em soundness} $(\alpha, m)$ if
\begin{align} \label{eq:soundness}
|\syndrome_X(e_X)| \leq \alpha |e_X|
\end{align}
for all reduced errors $e_X$ with weight $|e_X| \leq m$. 
In this article, we say that a family of CSS codes with length $n \rightarrow +\infty$ has good soundness -- or equivalently that it is locally testable -- if it has soundness $(\alpha, m)$ where $\alpha$ is a constant independent of $n$ and $m$ grows at least polynomially with $n$.

\section{Union-Find decoder for quantum LDPC codes} \label{sec:decoder}

In this section, we propose a Union-Find decoder for the correction of $X$ errors in arbitrary CSS codes, including quantum LDPC codes. The same strategy can be used to correct $Z$ errors by swapping the roles of $X$ and $Z$. We describe the algorithm in the Tanner graph \cite{tanner1981tanner_graph} associated with $Z$ checks. We focus on the description of the decoder and we ignore for now the implementation details and the complexity of the algorithm. 
In later sections we prove and demonstrate numerically that the Union-Find decoder defined here performs well for some important classes of LDPC quantum codes including topological and locally-testable codes.

\begin{figure}
\centering
\includegraphics[scale=.5]{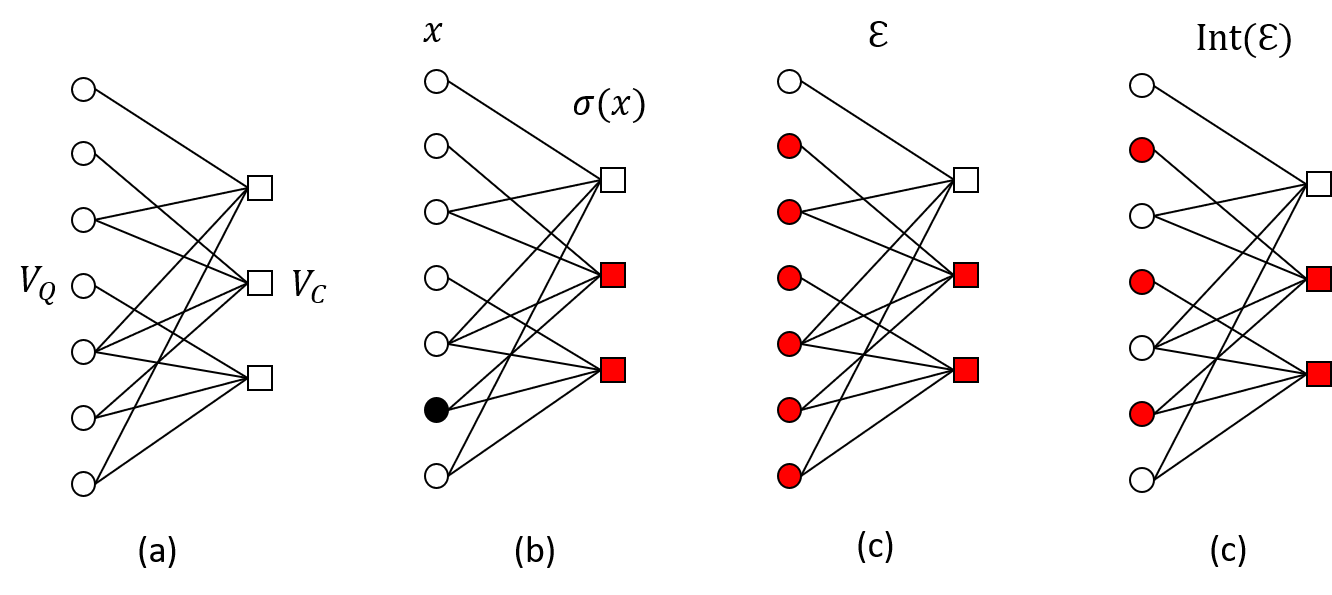}
\caption{
(a) Tanner graph of the Steane code. Circles are qubit nodes and squares are $Z$ check nodes. 
(b) A single-qubit error $\error$ supported on the black node and its syndrome $\syndrome$ supported on red nodes. 
(c) The set $\cE$ grown by Algorithm~\ref{algo:uf_decoder} after one growth steps. Initially $\cE$ contains the two red syndrome nodes. At the first growth steps, their neighbors are added to $\cE$.
(d) The interior of the set $\cE$.
The set $\cE$ is valid for the input syndrome. 
However, the decoder does not always succeed in this case because two non-equivalent valid corrections exist inside $\cE$.
}
\label{fig:tanner_graph}
\end{figure}

\medskip
The Tanner graph $T = (V, E)$ for $Z$ checks is a bipartite graph with vertex set $V = \vq \cup \vc$ where $\vq = \{q_1, \dots, q_n\}$ corresponds to the qubit set and $\vc = \{c_1, \dots, c_{r_Z}\}$ corresponds to the set of $Z$ checks. 
The nodes $q_i$ and $c_j$ are connected by an edge iff the check corresponding to $c_j$ acts non-trivially on qubit corresponding to $q_i$. 
The Tanner graph of the Steane code is represented in Fig.~\ref{fig:tanner_graph}(a).
 
\medskip
Here we find it convenient to represent an $X$ error $e_X \in \Z_2^n$ as the set of data nodes $\error \subseteq \vq$ that correspond to qubits in the support of $e_X$.
Similarly, we describe the syndrome $\syndrome(\error)$ as the set of check nodes that are incident to an odd number of nodes of $\error$.
In this article, we assume that the syndrome can be measured perfectly.

\medskip
The goal of a decoder is to identify the error $\error$ given its syndrome  $\syndrome(\error)$.
The basic idea of the Union-Find decoder is to reduce the correction of $X$ errors to the simpler problem of correcting an erasure. Recall that in the erasure channel model only erased qubits can suffer from an error. The location of the erased qubits is known and it is sufficient to identify a correction inside the erased set.

\medskip
The Union-Find decoder, described in Algorithm~\ref{algo:uf_decoder} works in two steps: First, we grow clusters of nodes around the non-trivial syndrome nodes, then we look for a correction inside these clusters. For a low weight error, if the grown clusters are small enough the correction is guaranteed to succeed. 
See Fig.~\ref{fig:tanner_graph} for an illustrative example.
Algorithm~\ref{algo:uf_decoder} provides the pseudo-code for our quantum LDPC code decoder.

\begin{algorithm}
	\caption{Union-Find decoder for quantum LDPC codes}
	\label{algo:uf_decoder}
	\vspace{0.2em}
	\textbf{Input:} The syndrome $\syndrome = \syndrome(\error) \subseteq \vc$ of an error $x \subseteq \vq$. \\
	\textbf{Output:} An estimation $\tilde \error \subseteq \vq$ of $\error$.
	\begin{algorithmic}[1]
		\STATE Initialize $\cE$ as $\cE = \syndrome$.
		\STATE While there exists an invalid connected component in $\cE$ do:
		\STATE \hspace{1cm} Add the neighbors of $\cE$ to $\cE$, that is $\cE \leftarrow B_T(\cE, 1)$.
		\STATE For each connected component $\cE_i$ of $\cE$ do:
		\STATE \hspace{1cm} Find a valid correction $\tilde \error_i \subseteq \Int(\cE_i)$ such that $\syndrome(\tilde \error_i) = \syndrome \cap \cE_i$.
		\STATE Return $\tilde \error = \cup_i \tilde \error_i$.
	\end{algorithmic}
\end{algorithm}

\medskip
A set $\cE \subseteq V$ of vertices of the Tanner graph is said to be {\em valid} for the syndrome $\syndrome$ if there exists an error $\tilde x \subseteq \vq \cap \Int(\cE)$ such that $\syndrome(\tilde x) = \syndrome \cap \cE$.
The condition $\tilde x \subseteq \vq \cap \Int(\cE)$ ensures that the syndrome $\syndrome(\tilde x)$ of this error is a subset of the set of check nodes of $\cE$.
If $\cE$ is valid for the syndrome $\syndrome$, an error $\tilde \error$ that satisfies the validity condition above is said to be a {\em valid correction} in $\cE$ for the syndrome $\syndrome$.
The following Lemma proves that the validity of a set depends on its connected components. 
Moreover, the decoder can compute a valid correction $\tilde \error$ by working independently on the connected components of $\cE$.

\begin{lemma} \label{lemma:validity_components}
Let $\cE \subseteq V$ and denote $\cE = \cE_1 \cup \dots \cup \cE_m$ the decomposition of $\cE$ into its connected components.
Then $\cE$ is valid iff all of its connected components are valid.
Moreover, consider an error $\tilde \error \subset \cE$.
Then $\tilde \error$ is a valid correction in $\cE$
iff
for all $i=1, \dots, m$, $\tilde \error \cap \cE_i$ is a valid correction in $\cE_i$.
\end{lemma}

\begin{proof}
It suffices to prove the equivalence:
$\tilde \error$ is a valid correction in $\cE$
iff
for all $i$, $\tilde \error \cap \cE_i$ is a valid correction in $\cE_i$.

Assume $\tilde \error$ is a valid correction in $\cE$ and let $\tilde \error_i = \tilde \error \cap \cE_i$.
Since $\tilde \error \subseteq \Int(\cE)$, we also have $\tilde \error_i \subseteq \Int(\cE_i)$.
Consider $\syndrome_i = \syndrome \cap \cE_i$.
Let $\syndrome = \syndrome(\tilde \error)$ and consider 
$\syndrome_i = \syndrome \cap \cE_i$.
By definition of the connected components, the only nodes of $\tilde \error$ that are at distance $\leq 2$ from a node of $\syndrome_i$ are in $\cE_i$, that is in $\error_i$. Moreover, $\sigma(\tilde \error_i) \subseteq \cE_i$. 
As a result, we have $\sigma(\tilde \error_i) = \sigma_i$.
This proves that $\tilde \error_i$ is a valid correction in $\cE_i$.
The converse implication can be proven with the same argument.
\end{proof}

\medskip 
Algorithm~\ref{algo:uf_decoder} relies on two subroutines:
\begin{itemize}
\item 
{\em The component validity subroutine} takes as an input a connected component of $\cE$ and the syndrome and it returns {\bf True} if the connected component is valid and {\bf False} otherwise.
\item 
{\em The component correction subroutine} takes as an input a valid connected component of $\cE$ and its syndrome and it returns a valid correction for the connected component.
\end{itemize}
For a general CSS code, these two subroutines can be reduced to solving a $\Z_2$-linear system of equations. 
This can be done in cubic complexity by Gaussian elimination, although in the case of the sparse matrices which arise for LDPC codes, a better complexity may be achievable \cite{mullen2013}. 
For surface codes, algorithms with almost-linear worst case complexity exist \cite{delfosse2020peeling}. 
In this work, we do not consider the optimization of these two subroutines for LDPC codes. Instead we focus on the error correction performance of the Union-Find decoder.

\medskip
The growth procedure of Algorithm~\ref{algo:uf_decoder} differs from the original Union-Find decoder because in the present work all connected components including the valid ones are grown simultaneously at each round. The proof of Proposition~\ref{prop:UF_success} relies on this modification of the growth procedure.

\medskip
The Union-Find decoder for CSS codes proposed in this article can be immediately generalized to arbitrary stabilizer codes \cite{gottesman1997stabilizer}. 
Each check node of the Tanner graph of a general stabilizer code corresponds to a stabilizer generator. 
The definition of valid connected components and valid corrections must be adapted by replacing the $X$ error $\error$ by a general Pauli error.

\section{Correction of errors with small covering radius} \label{sec:uf_success}

The Union-Find decoder introduced in Sec.~\ref{sec:decoder} can be applied to any CSS code. However, it is unlikely to succeed in correcting errors for a general code defined by large weight checks because the Union-Find clusters grow too fast and they will cover the whole Tanner graph after only a few growth rounds. One may expect to be confronted with the same issue for some LDPC codes, for example those with Tanner graphs which are expander graphs. However we will see that broad classes of errors can be corrected even for some families of LDPC codes with Tanner graphs exhibiting expansion. 

\medskip
In this section, we introduce the notion of a covering radius of an error and we prove that the Union-Find decoder successfully corrects all low-weight errors if their covering radius is small enough. 

\begin{defi}
The {\em covering radius of an error} $\error \subseteq \vq$ with syndrome $\syndrome$ is defined to be the minimum integer $r$ such that the $r$-neighborhood of $\syndrome$ in the Tanner graph $T$ covers $x$, that is
$$
\rcov(x) = \min \{ r \in \N \ | \ x \subseteq B_T(\syndrome, r)  \} \cdot
$$
If the syndrome of $x$ is trivial, we define $\rcov(x) = 0$.
\end{defi}

Fig.~\ref{fig:rcov_def} provides a visualization of the covering radius.
The {covering radius of a syndrome} $\syndrome \subset \vc$, denoted $\rcovbar(\syndrome)$, is defined to be the minimum covering radius of an error with syndrome $\syndrome$.

\begin{figure}
\centering
\includegraphics[scale=.5]{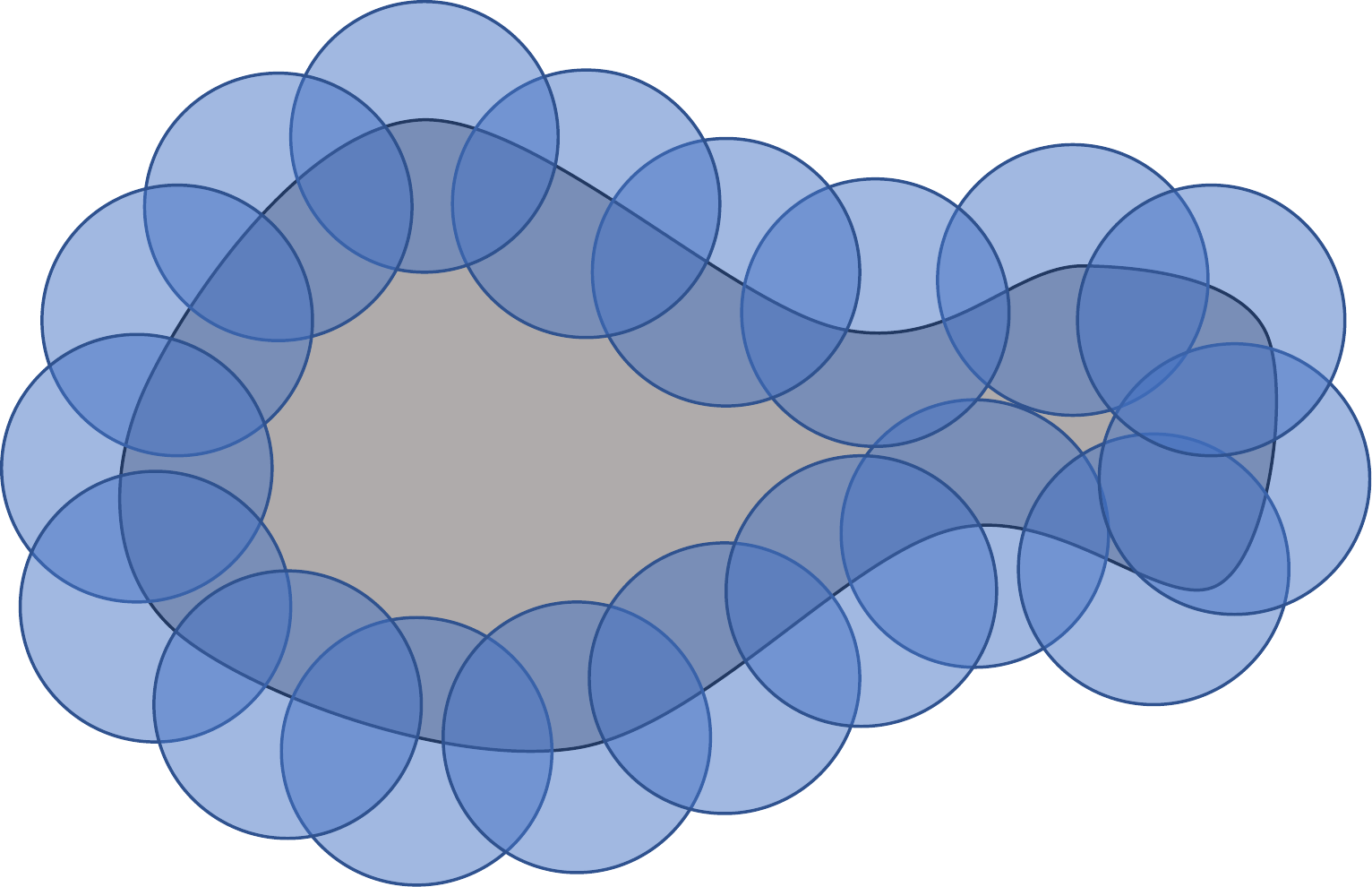}
\caption{Pictorial representation of an error $\error$ (grey) and its syndrome (the boundary of the grey region). To define the covering radius of this error, we consider the set of balls of radius $r$ centered at the syndrome nodes (blue region). The covering radius of $\error$ is the minimum radius $r$ such that the set of balls of radius $r$ covers the whole error $\error$.
}
\label{fig:rcov_def}
\end{figure}

\medskip 
The following proposition proves that the Union-Find decoder successfully corrects all errors with weight less than a polynomial in $d$ if their covering radii are small enough.
We will provide examples of standard families of quantum LDPC codes satisfying the covering radius condition for low-weight errors later.

\begin{prop} \label{prop:UF_success}
Consider a CSS code with minimum distance $d$ whose Tanner graph has bounded degree $\delta$. 
Assume that there exist constants $w$ and $C$ such that
$\rcovbar(\syndrome) \leq C \log(|\error|)$
for all errors $\error$ with syndrome $\syndrome$ and with weight $|x| < w$.
Then, the Union-Find decoder (Algorithm~\ref{algo:uf_decoder}) successfully corrects all errors with weight
$|\error| < \min(w, A d^\alpha)$
where 
$\alpha = \frac{1}{1+C\log(\delta)}$
and
$A = \left( \frac{1}{2\delta^2} \right)^{\alpha}$.
\end{prop}

\begin{proof}
Consider an error $\error$ such that $|\error| < \min(w, A d^\alpha)$ and let $\syndrome = \syndrome(\error)$.
By definition of the covering radius, we are guaranteed to obtain a set of valid connected components in Algorithm~\ref{algo:uf_decoder} after at most $\rcovbar(\syndrome)$ growth steps.
This proves that the correction is contained within
$
\cE \subseteq \bigcup_{s \in \syndrome} B_T(s, \rcovbar(\syndrome))
$, where the decomposition is over each vertex in the syndrome $\syndrome$ of $\error$.
Taking the cardinality on both sides, we obtain
$
|\cE| \leq |\syndrome| \cdot \delta^{1 + \rcovbar(\syndrome)}
$
where $\delta$ is the degree of the Tanner graph. 
Moreover, the weight of the syndrome is at most $\delta |\error|$. This proves that 
$$
|\cE| 
\leq \delta^2 |\error| \delta^{\rcovbar(\syndrome)} 
\leq \delta^2 |\error|^{1 + C \log(\delta)}.
$$
The Union-Find decoder returns a correction $\tilde \error$ included in $\cE$. After correction, the residual error $\error + \tilde\error$ (which has trivial syndrome) satisfies 
\begin{align*}
|\error + \tilde\error| 
\leq |\error| + |\cE| 
\leq 2 \delta^2 |\error|^{1 + C \log(\delta)}
< 2 \delta^2 \left( A d^\alpha \right)^{1 + C \log(\delta)}
= d.
\end{align*}\
This proves that the decoder succeeds since the residual error has weight $< d$, and therefore must be a stabilizer.
\end{proof}

\section{Covering radius of hyperbolic codes} \label{sec:hyp_codes}

Topological codes are among the most well-known quantum LDPC codes, and are defined in terms of qubits embedded in smooth manifolds, which are typically either hyperbolic or Euclidean; see Sec.~\ref{sec:manifolds} and Sec.~\ref{sec:LDPC-quantum-codes}.
The standard Union-Find decoder corrects any error with weight up to $(d-1)/2$ for two-dimensional hyperbolic and toric codes but it was not investigated in higher dimensions. 
In this section, we derive a bound on the covering radius of hyperbolic codes in dimension $D \geq 3$. Combined with Proposition~\ref{prop:UF_success}, this shows that a Union-Find decoder can successfully correct low-weight errors for hyperbolic codes in any dimension.

\medskip
Here, we consider a topological code defined on a cellulation of a $D$-dimensional hyperbolic manifold $\cM$ \cite{freedman2002z2, guth20144D_codes, londe2017golden_codes}, with the qubit set $Q = \cMi$ and the $Z$-check set $C = \cM_{i-1}$ following the notation in Sec.~\ref{sec:LDPC-quantum-codes}.
The syndrome of an $X$ error $\error \subseteq \cMi$ is $\partial(\error)$, the $\Z_2$-boundary of $\error$.
We focus here on the correction of $X$ errors. The same procedure corrects $Z$ errors in the dual cellulation. Indeed, the dual of the cellulation $\cM_0, \dots, \cM_D$ is the cellulation of $\cM$ obtained by replacing each $i$-cell by a $(D-i)$-cell.
Therefore, in the dual cellulation, qubits are placed on $(D-i)$-cells, $X$-checks correspond to $(D-i-1)$-cells and the syndrome of a $Z$ error is given by its $\Z_2$-boundary.

\medskip
Our upper bound on the covering radius exploits Anderson's theorem which provides a lower bound on the volume of a minimum $i$-chain in hyperbolic geometry. 
This result was also crucial to obtain a lower bound on the minimum distance of four-dimensional hyperbolic codes in the work of Guth and Lubotsky \cite{guth20144D_codes}.

\begin{theo} [Application of \cite{anderson1982min_varieties}, see \cite{guth20144D_codes} : Theorem 22 ] \label{theo:Anderson} 
Let $1 \leq i \leq D-1$ and let $X$ be an $i$-chain with boundary $\partial(X)$ in the hyperbolic space $\HypD$. 
Suppose that $X$ is a minimum $i$-chain with boundary $\partial(X)$.
Assume that $X$ contains the point $p \in \HypD$ and that its boundary $\partial(X)$
is included in the boundary of the ball $B_{\HypD}(p, R)$. Then
\begin{equation}
\label{eq:Anderson}
\Vol_i(X \cap B_{\HypD}(p, R)) \geq c_i e^{(i-1)R},
\end{equation}
for some constant $c_i$.
\end{theo}

The right hand side of Eq.~\ref{eq:Anderson} is proportional to the volume of the $i$-dimensional hyperbolic ball.

\begin{prop} \label{prop:cov_bound_hyperbolic}
Consider a hyperbolic topological code defined on a cellulation of a manifold $\cM$ with qubits on $i$-cells with $i \in \{2, \dots, D-1\}$ and $D \geq 3$. 
Let $w = \frac{c_i}{v_i'} e^{(i-1)(\Rinj - 1)}$ where $\Rinj$ is the injectivity radius of $\cM$, $c_i$ is the constant from Theorem~\ref{theo:Anderson} and $v_i'$ is the maximum volume of any $i$-cell in the cellulation.
For any error $\error$ with weight $|\error| < w$, the covering radius of the syndrome $\syndrome$ of $\error$ satisfies
$$
\rcovbar(\syndrome) \leq C \log(|\error|),
$$
for some constant $C$ that depends only on $D$, $i$, and $v_i'$. 
\end{prop}

\begin{proof}
To prove the proposition, we will derive a lower bound on $|\error|$ which is exponential in $\rcovbar(\syndrome)$ where $\syndrome$ is the syndrome of $\error$.
Given that our goal is to establish a lower bound on the weight of $\error$, we can assume that $\error$ is a minimum weight error with syndrome $\syndrome$. This minimality assumption will be necessary to apply Theorem~\ref{theo:Anderson}.

By definition of the covering radius, $B_T(\syndrome, \rcovbar(\syndrome)-1)$ does not cover the whole error $\error$ as one can see on Fig.~\ref{fig:proof}.
As a result, there exists some $q$ in $\error$ such that $q$ does not belong to $B_T(\syndrome, \rcov(\error)-1)$.
Then, we have 
$
B_T(q, \rcovbar(\syndrome)-1) \cap \syndrome = \emptyset.
$
We will obtain a lower bound on the weight of $\error$ by counting the qubits of $\error$ in this ball.
\begin{align*}
|\error| 
& \geq |\error \cap B_T(q, \rcovbar(\syndrome)-1)|, \\
& \geq \frac{1}{v_i'} \Vol_i(\error \cap B_\cM(q', \rcovbar(\syndrome)-1)).
\end{align*}
Therein, we used Eq~\eqref{eq:cells_vol} which provides an upper bound on the volume of cells and $q'$ is a point of $\cM$ included in the cell $q$.
Therein, the term $\error \cap B_\cM(q', \rcovbar(\syndrome)-1)$ refers to restriction of $\error$, considered as an $i$-chain, to the ball with center $q'$ and radius $\rcovbar(\syndrome)-1)$ in $\cM$.

Assume first that $\rcovbar(\syndrome) < \Rinj$.
Then, we can map $\error \cap B_\cM(q', \rcovbar(\syndrome)-1)$ onto an $i$-chain $y$ of $\HypD$ through the local homeomorphism $\pi$. 
By definition, $y$ is an $i$-chain included in a ball with radius $\rcov(\error)-1$ in $\HypD$ and its boundary $\partial(y)$ is included in the boundary of this ball.
Applying Theorem~\ref{theo:Anderson} we obtain
$
\Vol_i(y) \geq c_i e^{(i-1)(\rcovbar(\syndrome)-1)}.
$
Putting everything together this leads to 
\begin{align*}
|\error| \geq \frac{c_i}{v_i'} e^{(i-1)(\rcovbar(\syndrome)-1)}.
\end{align*}
Taking the log, this proves the proposition in the case
$\rcovbar(\syndrome) < \Rinj$.

If $\rcovbar(\syndrome) \geq \Rinj$, we can only use the local homeomorphism for radius $< \Rinj$.
This yields
\begin{align*}
|\error| 
\geq \frac{1}{v_i'} \Vol_i(\error \cap B_{\HypD}(q, \Rinj-1))
\geq \frac{c_i}{v_i'} e^{(i-1)(\Rinj-1)} = w,
\end{align*}
and there is nothing to prove since $|\error| \geq w$.
\end{proof}

\begin{figure}
\centering

\includegraphics[scale=.45]{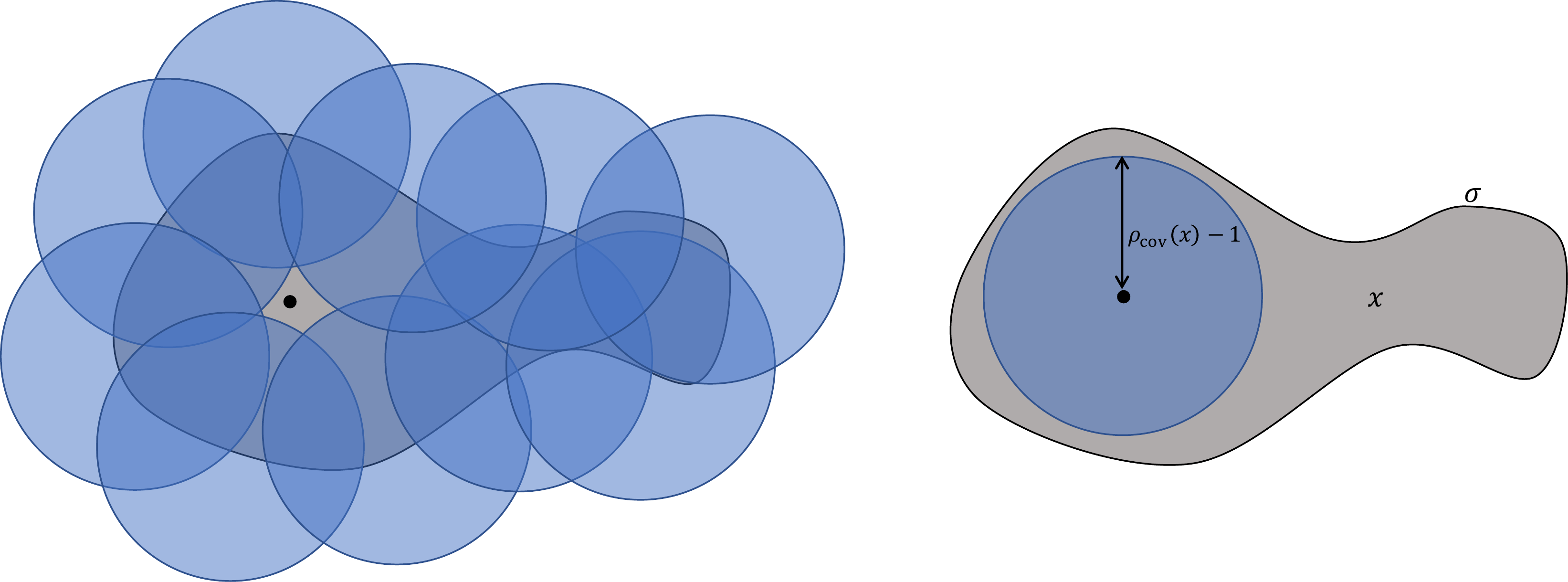}

(a) \hspace{6cm} (b)

\caption{
Illustration of the proof of Proposition~\ref{prop:cov_bound_hyperbolic} which uses the existence of a qubit $q$ in the support of the $\error$ with syndrome $\syndrome$ such that $B_T(q, \rcovbar(\syndrome))$ does not contain any node of $\syndrome$.
(a) By definition of the covering radius, the balls $B_T(s, \rcovbar(\syndrome)-1)$ (in blue) do not cover the error. Select $q \in \error$ in the non-cover region (black node).
(b) As a result the ball centered at $q$ with radius $\rcovbar(\syndrome)-1$ does not contain any syndrome node.
}
\label{fig:proof}
\end{figure}

\medskip
The following corollary, which is a straightforward application of Proposition~\ref{prop:UF_success} and Proposition~\ref{prop:cov_bound_hyperbolic}, is stated without a proof.

\begin{coro} \label{cor:hyperbolic_codes}
Consider a family of $D$-dimensional hyperbolic codes with qubits on $i$-cells with $i \in \{2, \dots, D-1\}$ and $D \geq 3$.
If the injectivity radius of the manifold grows at least logarithmically with the volume of the manifold, then the Union-Find decoder (Algorithm~\ref{algo:uf_decoder}) successfully corrects all errors with weight up to 
$Bn^\beta$ for some constants $B, \beta > 0$
\end{coro}

This corollary applies for instance to the 4D hyperbolic codes of Guth and Lubotsky~\cite{guth20144D_codes} based on the lower bound on the injectivity radius obtained in Corollary~11 in~\cite{guth20144D_codes}.
It also applies to the 4D hyperbolic codes studied by Londe and Breuckmann \cite{londe2017golden_codes}.

\section{Covering radius of locally testable codes} \label{sec:loc_testable_codes}

In this section, we derive a logarithmic upper bound on the covering radius from the soundness property of locally testable codes.
Our result applies to important classes of quantum LDPC codes such as quantum expander codes \cite{leverrier2015quantum_expander_codes}, and relies on the following proposition.

\begin{prop} \label{prop:cov_bound_soundness}
Let $\error$ be an error for a CSS code with a bounded degree Tanner graph and with soundness $(\alpha, m)$.
For any error $\error$ with weight $|\error| < m$, the covering radius of the syndrome $\syndrome$ of $\error$ satisfies
$$
\rcovbar(\syndrome) \leq C \log(|\error|),
$$
for some constant $C$ that depends only on $\alpha$ and the degree of the Tanner graph.
\end{prop}

\begin{proof}
The proof follows the same strategy as the proof of Proposition~\ref{prop:cov_bound_hyperbolic}. 
As we seek to prove a lower bound on the weight of $\error$, we can assume that it is a reduced error, i.e. it has the minimum weight given its syndrome $\syndrome$.
We start from a qubit $q \in \error$ such that 
$
B_T(q, \rcovbar(\syndrome)-1) \cap \syndrome = \emptyset,
$
and we use the lower bound
$
|\error| \geq |\error \cap B_T(q, \rcovbar(\syndrome)-1)|.
$
However, Anderson's theorem cannot be used in this context. We will derive a graphical version of Theorem~\ref{theo:Anderson}.

Let $R = \rcovbar(\syndrome)-1$. For each positive integer $k \leq R$, define $\error_k = \error \cap B_T(q, k)$ and let $\syndrome_k = \syndrome(\error_k)$. We can restrict our attention to even values of $k$ since the Tanner graph is bipartite (and therefore $\error_{2j} = \error_{2j-1}$). 
From the soundness property we have 
\begin{align} \label{eq:soundess_s_lower_bound}
|\sigma_{k}| \geq \alpha |\error_{k}| \cdot
\end{align}
Note that $\error$ having minimum weight given $\syndrome(\error)$ implies that $\error_{k}$ has minimum weight given $\syndrome_{k}$. 
The soundness inequality applies because $|\error_k| \leq |\error| \leq m$ by assumption.
Let $y_k = \error \cap S_T(q, k)$ be the intersection of the error with the sphere with radius $k$. 
By definition, $\error_{k+2}$ is the disjoint union of $\error_{k}$ and $y_{k+2}$. Therefore $|\error_{k+2}|=|\error_{k}|+|y_{k+2}|$. Let us show that if $k \leq R$, then
\begin{align} \label{eq:soundess_y_lower_bound}
|y_{k+2}| \geq \frac{1}{\delta} |\syndrome_{k}|,
\end{align}
where $\delta$ is the degree of the Tanner graph.

To prove Eq.~\eqref{eq:soundess_y_lower_bound}, we will first show that $\syndrome_k$ is included in the sphere $S_T(q, {k+1})$. Clearly, $\syndrome_k = \syndrome(\error_k)$ is included in the ball $B_T(q, {k+1})$ because $\error_k \subset B_T(q, {k})$.
We will now prove by contradiction that there is no vertex $s \in \syndrome_{k}$ included in $B_T(q, {k})$. 
Assume that such a vertex $s \in \syndrome_k \cap B_T(q, {k})$ exists. 
Because all the vertices at distance $k$ from $q$ are qubit vertices, we actually have $s \in \syndrome_k \cap B_T(q, k-1)$.
Consider now the syndrome of $x_{k+2}$.
Given that $x_{k+2}$ is the disjoint union of $x_{k}$ with a set $y_{k+2}$ of vertices at distance $k+2$ from $q$, the vertex $s$ is at distance $\geq 2$ from $y_{k+2}$. This proves that $s$ must belong to $\syndrome_{k+2} = \syndrome(x_{k+2})$.
By induction, $s$ also belongs to $\syndrome$. This is a contradiction since $k \leq R$ and it thus concludes our proof by contradiction.

The next step to proving Eq.~{\eqref{eq:soundess_y_lower_bound}} is to note that for $k \leq R-1$, the fact that any node $s \in \syndrome_k$ is a distance $k+1$ from $q$, implies that $s \not\in \syndrome$.
In order for $s$ to be absent in $\sigma$ but present in $\sigma_k$, there must be at least one qubit in $y_{k+2}$ which is adjacent to $s$ to remove it from the syndrome.
Moreover, each qubit of $y_{k+2}$ is adjacent to at most $\delta$ nodes of $\syndrome_k$. Hence $|y_{k+2}| \geq |\syndrome_k|/\delta$ proving Eq.~{\eqref{eq:soundess_y_lower_bound}}.

Combining ~\eqref{eq:soundess_s_lower_bound} and \eqref{eq:soundess_y_lower_bound}, and using the fact that $|\error_{k+2}|=|\error_{k}|+|y_{k+2}|$, we find
\begin{align} \label{eq:soundness_expansion}
|\error_{k+2}| \geq \left( 1 + \frac{\alpha}{\delta} \right) |\error_k| \cdot
\end{align}
This is valid for all $k \leq R = \rcovbar(\syndrome) - 1$. As a result, we get
$
|\error| \geq |\error_R| \geq B e^{\beta R}
$
for some constants $B, \beta > 0$, proving the Proposition.
\end{proof}

Proposition~\ref{prop:cov_bound_soundness} applies to standard families of quantum LDPC codes such as expander codes because they have $(\alpha >0, \Omega(\sqrt{n}))$-soundness (See Corollary~9 in~\cite{leverrier2015quantum_expander_codes}).

\begin{coro} \label{cor:expander_codes}
The Union-Find decoder (Algorithm~\ref{algo:uf_decoder}) successfully corrects all errors with weight up to $Bn^\beta$ for some constants $B, \beta > 0$
for the quantum expander codes of \cite{leverrier2015quantum_expander_codes}
\end{coro}

Our proof technique leads to a weaker inequality than Proposition~\ref{prop:cov_bound_soundness} for codes that are only logarithmically LDPC and have a logarithmically good soundness such as those defined in \cite{hastings2016quantum} and \cite{leverrier2019towards}.

\section{Correctability bound for toric codes}
\label{sec:toric_codes}

This section focuses on the $D$-dimensional toric codes defined on the torus $\TD = \R^D / (L\Z)^D$ tiled by $L^D$ cubic cells with sides of length one.
We consider the toric code $\TC(D, i)$ obtained by placing qubits on the $i$-cells of $\TD$.
We prove that the Union-Find decoder corrects any error with weight up to $C L^{i/2}$ for some constant $C$.

Since the covering radius of toric codes does not satisfy the requirements of Proposition~\ref{prop:UF_success}, we need to develop a different proof technique. Our argument relies on a bound on the covering radius and the disjointness of the set of logical errors for toric codes \cite{jochym2018disjointness, webster2020universal}.

\subsection{Covering radius of toric codes}

Here we use the same techniques as the bound on the covering radius of hyperbolic codes obtained in Section~\ref{sec:hyp_codes}.

\medskip
The Euclidean analog of Theorem~{\ref{theo:Anderson}} leads to a lower bound on minimum $i$-chains of the form
\begin{align} \label{eq:anderson_euclidean}
\Vol_i(X \cap B_{\R^D}(0, R)) \geq c_i R^{i},
\end{align}
rather than $c_i e^{(i-1)R}$ in the hyperbolic case.
The volume of the Euclidean ball (on the right hand side) is smaller than the volume of the hyperbolic ball with same radius.
Based on Eq.~{\eqref{eq:anderson_euclidean}}, one can establish the Euclidean version of Proposition~{\ref{prop:cov_bound_hyperbolic}}.

\begin{prop} \label{prop:cov_bound_euclidean}
Consider the toric code $\TC(D, i)$ with $i \in \{2, \dots, D-1\}$ and $D \geq 3$ defined on the torus with length $L$.
Let $w = \frac{c_i}{v_i'}(L-1)^{i}$ where $c_i$ is the constant from Eq.~\eqref{eq:anderson_euclidean} and $v_i'$ is the maximum volume of any $i$-cell in the cellulation.
For all error $\error$ with weight $|\error| < w$, the covering radius of the syndrome $\syndrome$ of $\error$ satisfies
$$
\rcovbar(\syndrome) \leq C |\error|^{\frac{1}{i}},
$$
for some constant $C$ that depends only on $D$, $i$, and $v_i'$. 
\end{prop}

\medskip
The covering radius of syndromes in toric codes is too large to apply Proposition~\ref{prop:UF_success}. 
In the rest of this section, we present a different argument based on the shape of logical errors.

\subsection{Logical errors for CSS codes}

An error $e_X$ for a CSS code with parity-check matrices $H_X$ and $H_Z$ is called a {\em logical $X$ error} if its syndrome is trivial.
Equivalently, a logical $X$ error is an error of $\Ker H_Z$.
A logical error that is a stabilizer is said to be a {\em trivial logical error}, otherwise it is said to be non-trivial. 
In other words, a non-trivial logical $X$ error is an error of 
$\Ker H_Z \backslash (\Ker H_X)^\perp$.

\medskip
By a {\em basis of logical $X$ errors}, we mean a family of $k$ non-trivial logical errors $\ell_X^{(1)}, \dots, \ell_X^{(k)}$ that generate the quotient space 
$
\Ker H_Z / (\Ker H_X)^\perp
$
which is the space of logical $X$ errors modulo the $X$ stabilizers.
Logical $Z$ errors are defined in the same way by swapping the roles of $X$ and $Z$.

\medskip
We can use a basis of logical errors to check whether a logical error is trivial.

\begin{lemma} \label{lemma:logical_error_test}
Let $\ell_Z^{(1)}, \dots, \ell_Z^{(k)}$ be a basis of logical $Z$ errors.
A logical $X$ error $e_X$ is non-trivial iff for all $j=1, \dots, k$, we have
$(e_X, \ell_Z^{(j)}) = 0 \pmod 2$.
\end{lemma}

\begin{proof}
Since it is a logical operator, we have $e_X \in \Ker H_Z$.
It is trivial iff it is a stabilizer, that is, iff it belongs to $(\Ker H_X)^\perp$.
By definition, $(\Ker H_X)^\perp$ is the orthogonal complement of the space $\Ker H_X$ which is the space of logical $Z$ errors.
Therefore, it is enough to check that $e_X$ is orthogonal with a basis of this space.
\end{proof}

\subsection{Logical errors for toric codes}

Recall that an $X$ error for the toric code $\TC(D, i)$ is an $i$-chain $x$ and its syndrome is the $(i-1)$-chain $\sigma = \partial(x)$.
An error is a stabilizer iff it is the boundary $x = \partial(w)$ of an $(i+1)$-chain $w$.
An $X$ error $x$ is a logical error if its boundary is trivial and this logical error is non-trivial if it is not a boundary. 
In the language of homology theory, a logical $X$ error is a non-trivial $i$-cycle \cite{hatcher2005algebraic_top}.

\medskip
To describe logical errors for the toric code, we introduce some notation. 
Let $e_1, \dots, e_D$ be the vectors of the standard basis of $\R^D$.
Given a point $p \in \TD$ and a set $I \subseteq \{1, \dots, D\}$, define
\begin{align} \label{def:affine_space}
\cA(p, I) 
= \{ p + \sum_{j \in I} \lambda_j e_j \ | \ \lambda_j \in [0, L] \}.
\end{align}
It is an $i$-dimensional section of the torus with $i = |I|$.
In what follows, we denote $I^C = \{ 1, \dots, D \} \backslash I$.
One can check that $\cA(p, I)$ and $\cA(p', I^C)$ intersect at a single point.
For instance, in the 3D torus ${\mathbb T}^3_{L}$ the set $\cA(p, \{1\})$ is a line oriented in the direction of $e_1$ and $\cA(p, \{2, 3\})$ is a plane with direction $e_2, e_3$. The intersection of these two subsets is the point $p$.

\medskip
Any $i$-cell of the torus is of the form 
$
c(p, I) = \{ p + \sum_{j \in I} \lambda_j e_j \ | \ \lambda_j \in [0, 1] \}
$
for some set $I \subseteq \{1, \dots, D\}$ with size $i$ and where $p \in \Z_L^D$.
Define the $i$-chains
\begin{align} \label{def:lX}
\ell_X(p, I) & = \sum_{q \in \cA(p, I) \cap \Z_L^D} c(q, I), \\
\ell_Z(p, I) & = \sum_{q \in \cA(p, I^C) \cap \Z_L^D} c(q, I) \cdot
\end{align}
The $i$-chain $\ell_X(p, I)$ forms a tiling of the $i$-dimensional sheet $\cA(p, I)$ of the torus and $\ell_Z(p, I)$ is the dual of a $(D-i)$-dimensional sheet of the torus.

\medskip
Examining the relative positions of the errors $\ell_X$ and $\ell_Z$, we see that they satisfy the following relations where $p, p' \in \Z_L^D$.
\begin{itemize}
\item $\ell_X(p, I) \cap \ell_X(p', I) \neq \emptyset$ iff $p' \in \cA(p, I)$. Moreover, if this condition is satisfied, we have $\ell_X(p, I) = \ell_X(p', I)$.
\item $\ell_Z(p, I) \cap \ell_Z(p', I) \neq \emptyset$ iff $p' \in \cA(p, I^C)$. Moreover, if this condition is satisfied, we have $\ell_Z(p, I) = \ell_Z(p', I)$.
\item
$|\ell_X(p, I) \cap \ell_Z(p', I)| = 1$.
\end{itemize}

\medskip

\begin{prop} [logical basis] \label{prop:toric_code_log_errors}
For any $p \in \Z_L^D$, the $\binom{D}{i}$ errors $\ell_X(p, I)$ (respectively $\ell_Z(p, I)$) where $I$ is a subset of $\{1, \dots, D\}$ with size $i$, form a basis for the space of logical $X$ errors (respectively logical $Z$ errors).

Moreover, the logical errors $\ell_X(p, I)$ and $\ell_X(p', I)$ (respectively $\ell_Z(p, I)$ and $\ell_Z(p', I)$) are equivalent up to a stabilizer.
\end{prop}

\begin{proof}
It is an immediate consequence of the study of the homology of the torus. 
This is because non-trivial logical $X$ errors correspond to non-trivial $i$-cycles. The same argument applies for $Z$ errors in the dual cellulation.
\end{proof}

In other words, each set $I$ with size $i$ can be used as the index of a logical error of this basis. 
An $X$ logical error $\ell_X(p, I)$ contains $L^{i}$ cells and there are $L^{D-i}$ disjoint parallel representatives of this logical operator which are equivalent to one another up to a stabilizer.
The $L^{D-i}$ disjoint representatives of $\ell_X(p, I)$ are the errors $\ell_X(q, I)$ with $q \in \cA(p, I^C)$.
A $Z$ logical error $\ell_Z(p, I)$ contains $L^{D-i}$ cells and has $L^{i}$ disjoint parallel representatives.
The $L^{i}$ disjoint representatives of $\ell_Z(p, I)$ are the errors $\ell_Z(q, I)$ with $q \in \cA(p, I)$.

\medskip
We can use this logical basis to detect non-trivial logical errors.

\begin{lemma} \label{lemma:non_trivial_error_test}
If $x$ is a non-trivial logical error for the code $\TC(D, i)$, then there exists a set $I \subseteq \{1, \dots, D\}$ with size $i$ such that 
$
x \cap \ell_Z(p, I) \neq \emptyset
$
for all $p \in \Z_L^D$.
\end{lemma}

\begin{proof}
This is the translation of Lemma~\ref{lemma:logical_error_test} in the language of $i$-chains.
The binary inner product between two error vectors corresponds to the intersection size 
$x \cap z \mod 2$ 
of the two $i$-chains $x$ and $z$ representing the errors.
\end{proof}

\subsection{Correctability bound for toric codes}

\begin{prop}
The Union-Find decoder (Algorithm~\ref{algo:uf_decoder}) successfully corrects all errors with weight up to $C L^{i/2}$ for some constant $C$ that depends only on $i$.
\end{prop}

\begin{proof}
Assume that the decoder fails for an error $\error$ and consider the residual error after correction $\tilde \error + \error$ where $\tilde \error$ is the proposed correction returned by the decoder.

For an error $u$, denote $n_I(u)$ the number of disjoint logical errors $\ell_Z(q, I)$ with varying $q \in \Z_L^D$ that meet $u$.
By definition of $n_I$, we have 
\begin{align}
|\error| \geq n_I(\error).
\end{align}
In what follows, we will derive a lower bound on $n_I(\error)$.

Consider first $n_I(y)$ where $y = \tilde \error + \error$.
Given that the decoder fails, the residual error $y$ is a non-trivial logical $X$ error.
Based on Lemma~\ref{lemma:non_trivial_error_test}, there exists a set $I$ with size $i$ such that each of the $L^i$ distinct logical errors 
$
\ell_Z(q, I)
$
with $q \in \cA(p, I)$ meets $y$. 
This shows that $n_I(y) = L^i$.
Then, applying a union bound, we get
$
n_I(y) \leq n_I(\error) + n_I(\tilde \error)
$
which leads to
\begin{align}
n_I(\error) \geq L^i - n_I(\tilde \error) \cdot
\end{align}

We will now establish an upper bound on $n_I(\tilde \error)$.
The correction $\tilde \error$ is included in the set $\varepsilon$ grown by the decoder, which implies that 
$
n_I(\tilde \error) \leq n_I(\varepsilon).
$
Let us bound $n_I(\varepsilon)$.
We know that $\varepsilon \subseteq B_T(\syndrome, \rcovbar(\syndrome))$ where $\syndrome$ is the syndrome of $\error$.
By a union bound, this leads to
\begin{align}
n_I(\varepsilon) \leq \sum_{s \in \syndrome} n_I(B_T(s, \rcovbar(\syndrome))).
\end{align}
The $i$-cells of the ball $B_T(s, R)$ of the Tanner graph are included in the ball $B_{\TD}(p_s, b_i(R+1))$ of the torus where 
$b_i = \sqrt{i}$ and where $p_s$ is a point of $\Z_L^D$ that belongs to the cell $s$.
This is because the diameter of an $i$-cell is $b_i$.
Combining this with Lemma~\ref{lemma:intersection_bound}, we get
$$
n_I(B_T(s, R)) 
\leq n_I(B_{\TD}(p_s, b_i (R+1))) 
\leq  i^{i/2} a_i (R+1)^i.
$$
Overall, we proved that
\begin{align}
n_I(\tilde \error)
\leq
n_I(\varepsilon) 
\leq 
|\syndrome|  i^{i/2} a_i(\rcovbar(\syndrome)+1)^i,
\end{align}
for some constant $a_i$.

Putting everything together, we showed that
\begin{align} \label{eq:proof_eucl_error_bound}
|\error| 
\geq 
L^i - |\syndrome| i^{i/2} a_i (R+1)^i 
= 
L^i - |\syndrome| B \rcovbar(\syndrome)^i,
\end{align}
for some constant $B$ that depends only on $i$.
Injecting the bound on the covering radius from Proposition~\eqref{prop:cov_bound_euclidean} in this equation, we find
\begin{align} \label{eq:proof_eucl_error_bound_2}
|\error|
& \geq L^i - C |\error|^2,
\end{align}
for some constant $C$ that depends only on $i$.
To bound the syndrome weight, we used the bound 
$|\syndrome| \leq \delta |\error|$ where $\delta$ is the degree of the Tanner graph. The constant $\delta$ depends only on $i$.

Eq.~\eqref{eq:proof_eucl_error_bound_2} provides a lower bound on the weight of errors that lead to a decoder failure.
If an error is uncorrectable, then its weight $w$ must satisfy
$
w + Cw^2 \geq L^i.
$
Consider an error with weight $w \leq \sqrt{1/(C+1)} L^{i/2}$. 
Then, we have 
$
w + Cw^2 \leq (1+C)w^2 < L^i,
$
violating the uncorrecability condition.
This proves that any error with weight $w \leq \sqrt{1/(C+1)} L^{i/2}$ is corrected successfully by the Union-Find decoder.
\end{proof}

\begin{lemma} \label{lemma:intersection_bound}
Let $I$ be a fixed subset of $\{1, \dots, D\}$ with size $i$.
Let $b = B(p_0, R)$ be a ball with center $p_0 \in \Z_L^D$ and radius $R < L$ in $\TD$. 
Denote $n_I(b)$ the number of distinct errors $\ell_Z(q, I)$ with varying $q \in \Z_L^D$ that meet $b$.
We have 
$$
n_I(b) \geq a_i R^{i},
$$
for some constant $a_i$.
\end{lemma}

\begin{proof}
By transitivity of the torus' cellulation, we can assume that $p_0 = 0 \in \Z_L^D$.
Let $\ell_Z(q, I)$ be an error that meets $b$ and let $c(q', I)$ be a $i$-cell of this intersection. 

We will make use of the orthogonal projection $\pi$ that maps a point
$
\sum_{j=1}^D \lambda_j e_j
$ 
of the ball onto the point
$
\sum_{j \in I} \lambda_j e_j.
$
The projection $\pi$ maps a point onto $\cA(0, I)$ and it kills its component in $\cA(0, I^C)$.

The image of the cell $c(q', I)$ under $\pi$ is the cell $c(\pi(q'), I)$ and it satisfies
\begin{align*}
c(\pi(q'), I) \in \ell_X(0, I) \cap \ell_Z(q, I) \cdot
\end{align*}
Indeed, $c(\pi(q'), I)$ is included in $\ell_X(0, I)$ because $\pi(q') \in \cA(0, I)$, and it belongs to $\ell_Z(q', I)$ because $q'$ and $\pi(q')$ differ by a vector of $\cA(0, I^C)$.

We proved that if $\ell_Z(q, I)$ meets the ball $b$, it must also meet $\ell_X(0, I)$ inside this ball. 
Moreover, we know that two distinct errors $\ell_Z(q, I)$ and $\ell_Z(q', I)$ do not intersect.
This yields
\begin{align*}
n_I(b) \geq |\ell_X(0, I) \cap b| \geq C R^i,
\end{align*}
for some constant $C$.
Therein, the last inequality is given by the size of a $i$-dimensional ball (it is the intersection of the $i$-dimensional space with a ball of radius $R$).
\end{proof}

\section{Numerical results}
\label{sec:numerics}

In previous sections we have shown that our Union-Find decoder performs well for a variety of code families against adversarial noise by proving it corrects arbitrary low-weight errors. 
In this section we provide numerical evidence that our Union-Find decoder also performs well in the setting of independent Pauli noise. 
In particular, we simulate a 4D hyperbolic topological code which encodes 736 logical qubits in 3600 qubits and compare the logical failure probability when decoded with our Union-Find decoder with that obtained using a BP decoder as in \cite{breuckmann2020single_shot_4D_codes}. 
We refer to Section III  of \cite{breuckmann2020single_shot_4D_codes} for a comprehensive description of the construction of this quantum code and other instances that belong to the family of 4-dimensional hyperbolic codes. The quantum code we use in this paper corresponds to the $4^{th}$ row of Table 1 of \cite{breuckmann2020single_shot_4D_codes}.
In \cite{breuckmann2020single_shot_4D_codes}, the BP decoder exhibited a high threshold and good single-shot behavior. 
In this work we investigate the low error rate regime and show numerically that the Union-Find decoder outperforms the BP decoder for error rates less than $5 \cdot 10^{-4}$; see Fig.~\ref{fig:numericsHyperbolic}.
The details of the BP decoder are described in Appendix \ref{appendix}.

\medskip
We consider the independent Pauli noise model with perfect syndrome extraction. Since the code we consider is symmetric with respect to X and Z errors, we only consider Z errors.

\begin{figure}[H]
\centering
\includegraphics[width=0.6\textwidth]{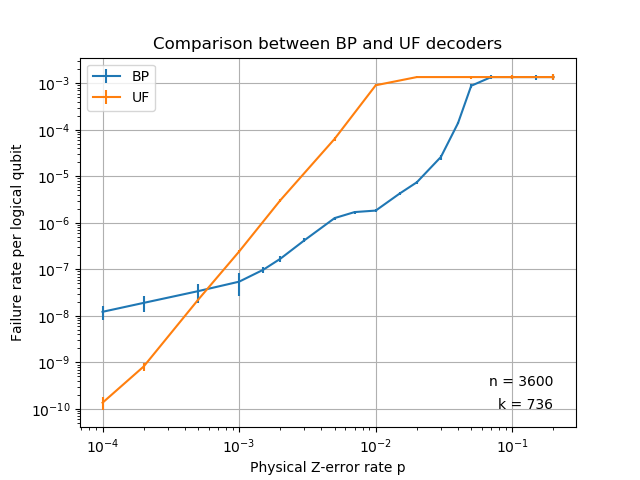}
\caption{The logical failure probabilities per logical qubit of a 4D hyperbolic topological code when decoding with a BP decoder (blue) and the Union-Find decoder (orange).
The code instance used encodes 736 logical qubits in 3600 qubits and is specified in detail in \cite{breuckmann2020single_shot_4D_codes}.}
\label{fig:numericsHyperbolic}
\end{figure}

\medskip
In Fig.~\ref{fig:numericsHyperbolic} we divide the logical error probability by the number of logical qubits (736 in our case). 
This allows one to meaningfully compare the performance of different codes which encode different numbers of logical qubits.
The justification for this normalization factor can be understood in the scenario of a choice between storing $k$ logical qubits in a set of $k$ copies of a code which encodes a single logical qubit and a single code that encodes $k$ logical qubits.
The probability $p_{\text{success}}$ of no failures occurring in any of the $k$ copies of a code which each fail independently with probability $p_\text{L}$ satisfies:
$$
p_{\text{success}} = 1 - (1 - p_\text{L})^k \approx k p_\text{L},
$$  
where the approximation holds for small $p_{\text{success}}$.
Dividing by $k$ therefore yields the logical error rate that a code with one logical qubit would need to perform the same computation.

\medskip
At low error rates, the Union-Find decoder exhibits a slope that is greater than 3 on the log-log plot of logical error rate versus physical error rate. This is a strong indication that the effective minimum distance of this code equipped with the Union-Find decoder is at least 5. Indeed if all errors of weight less than or equal to $\lfloor \frac{d - 1}{2} \rfloor$ are corrected, the dominant source of errors at low error rates comes from errors of weight $\lfloor \frac{d - 1}{2} \rfloor + 1$. Hence the slope of log-log plot. 

\medskip
Interestingly the BP decoder doesn't exhibit the same behavior at low error rates.
Moreover the non-increasing behavior of the derivative of the performance curve of the BP decoder is surprising at first sight. A possible explanation comes from the use of the physical error rate as an input of the BP decoder. Therefore it is possible that a given error is succesfully corrected at low error rates but leads to a failure at a higher error rate.
The fact that the BP decoder performs poorly with low errors that cover part of a check was already observed in \cite{poulin2008iterative} with a different variant of BP decoder.

\medskip
At sufficiently low error rates, the Union-Find decoder performs better than the BP decoder. 
The crossing point is around $5 \cdot 10^{-4}$ for the code we simulated. These numerical results are consistent with our theoretical analysis which showed that errors of weight less than a power of the minimum distance are always corrected by the Union-Find decoder.

\medskip
The failures of the BP decoder are of two kinds. The first kind consists of logical errors. The second kind happens when the decoder doesn't find any error corresponding to the observed syndrome. The second kind of failure is less problematic for error correction since the decoder knows that it is failing : we can think of them as flagged errors. Interestingly, in our simulations, most of the BP decoder failures are flagged errors. A straightforward combination of the BP and Union-Find decoders consists in running the BP decoder first and running the Union-Find decoder when the BP decoder outputs a flagged error. We expect such a combination to give a lower failure rate than any of the two decoders alone and leave a numerical investigation for future work.

\section{Conclusion}

We described a generalization of the Union-Find decoder and we proved that it corrects any error with weight up to $A n^\alpha$ for some constant $A$, for three standard classes of quantum LDPC codes.
The main difference with the standard Union-Find decoder is that we keep growing valid clusters. This is because these clusters help linking invalid clusters, making them valid. As a result the final grown set of qubits may be smaller in the case of expander codes or hyperbolic codes for which the size of a ball of the Tanner graph grows exponentially fast with the radius.

\medskip
Our decoder could also be applied to the recently proposed fiber bundle LDPC codes~\cite{hastings2020fiber, breuckmann2020balanced} and the product codes of~\cite{panteleev2020quantum} but we did not investigate its performance for these codes.
The generalization of the Union-Find decoder to the homological product of a topological codes and a small block code was considered in \cite{delfosse2020UFproduct}.

\medskip
It would be interesting to explore different variants of the growth. In the case of 2D surface codes, growing one cluster at a time and growing the smallest cluster first is beneficial~\cite{delfosse2017UF_decoder}. One can also consider growing a single qubit at a time~\cite{Hu2020quasi} or taking into account edge weights to orient the growth~\cite{huang2020fault}.

\medskip
We compared numerically the performance of our Union-Find decoder the BP decoder. In the future, it would valuable to explore numerically the performance of our decoder for different families of LDPC codes and to compare the results with other variants of the BP decoders \cite{panteleev2019BPOSD, roffe2020decoding, grospellier2018numerical, grospellier2020combining} or the recent linear programming decoder of \cite{fawzi2021LP}.

\medskip
In this work, we have focused on the error correcting performance of the decoder and while our algorithm runs in polynomial time, we leave the question of optimizing its computational complexity for later. 
The bottleneck of our implementation is the use of Gaussian elimination for the component validity subroutine and the component correction subroutine.
Like for the standard Union-Find decoder,
it may be possible to achieve an almost-linear complexity by adapting the peeling decoder~\cite{delfosse2020peeling}.

\medskip
We leave the question of proving the existence of the threshold open. 
We believe that our decoder could be made fault-tolerant for quantum LDPC codes and especially locally testable codes which are naturally robust against syndrome measurement noise \cite{campbell2019single_shot}.

\section{Acknowledgments}

We would like to thank Nikolas Breuckmann for providing the hyperbolic topological code used for numerics and Jeongwan Haah and Anthony Leverrier for their comments on a preliminary version of this article.


\section{Appendix: Belief Propagation decoder}
\label{appendix}

The basic idea of the standard BP decoder \cite{gallager1962LDPC} is to estimate the marginal error probability of a qubit using the value of the syndrome over all the checks at distance less than $2R$ in the Tanner graph \cite{richardson2008modern}. 

\medskip
This marginal error probability for a qubit can be estimated efficiently through rounds of message passing during which each node of the Tanner graph exchanges information with its neighbors. After $R$ rounds of message passing (a rounds consists of messages from qubits to checks and messages from checks to qubits, which explains the factor $2$ discrepancy), each qubit has accumulated data from all the checks at distance less than $2R$ from itself. Based on this information, one can determine if a qubit is more likely to support an error or not.

\medskip
The radius $2R$ must be optimized for each error correcting code and for each physical error rate probability, which is tedious and can lead to unfair comparisons. In this work, we consider the variant of the BP decoder used in \cite{breuckmann2020single_shot_4D_codes} which does not require to tune the parameter $R$.

\begin{algorithm}
	\caption{BP with $R$ rounds}
	\label{algo:bp_decoder_R_rounds}
	\vspace{0.2em}
	\textbf{Input:} The syndrome $\syndrome = \syndrome(\error) \subseteq \vc$ of an error $x \subseteq \vq$. The physical error rate $p$. \\
	\textbf{Output:} An estimation $\tilde \error \subseteq \vq$ of $\error$.
	\begin{algorithmic}[1]
		\STATE \underline{Initialization:}
		\STATE Initialize $\tilde \error = \varnothing$
		\STATE For each pair $(q, c)$ of a qubit and a check that are adjacent do:
		\STATE \hspace{1cm} Initialize a value from qubit to check: $\ell_{q \rightarrow c} =\log \frac{1-p}{p} $
		\STATE \underline{$(R-1)$ first rounds:}
		\STATE For $R-1$ rounds do:
		\STATE \hspace{1cm} For each pair $(q, c)$ of a qubit and a check that are adjacent do:
		\STATE \hspace{2cm} a = 0
		\STATE \hspace{2cm} For each check $c_2$ adjacent to $q$ and different from $c$ do:
		\STATE \hspace{3cm} $b = 1$
		\STATE \hspace{3cm} For each qubit $q_2$ adjacent to $c_2$ and different from $q$ do:
		\STATE \hspace{4cm} $b \leftarrow b * \tanh( \ell_{q_2 \rightarrow c_2} )$
		\STATE \hspace{3cm} Value from check to qubit: $ \ell_{c_2 \rightarrow q} = \frac{(-1)^{c_2 \in \sigma}}{2} \atanh( b ) $
		\STATE \hspace{2cm} $a \leftarrow a + \ell_{c_2 \rightarrow q}$
		\STATE \hspace{2cm} Value from qubit to check: $\ell_{q \rightarrow c} = \log \frac{1-p}{p}+a$.
		\STATE \underline{Last round:}
		\STATE For each qubit $q$ do:
		\STATE \hspace{1cm} a = 0
		\STATE \hspace{1cm} For each check $c_2$ adjacent to $q$ do:
		\STATE \hspace{2cm} $b = 1$
		\STATE \hspace{2cm} For each qubit $q_2$ adjacent to $c_2$ and different from $q$ do:
		\STATE \hspace{3cm} $b \leftarrow b * \tanh( \ell_{q_2 \rightarrow c_2} )$
		\STATE \hspace{2cm} Value from check to qubit: $ \ell_{c_2 \rightarrow q} = \frac{(-1)^{c_2 \in \sigma}}{2} \atanh( b ) $
		\STATE \hspace{1cm} $a \leftarrow a + \ell_{c_2 \rightarrow q}$
		\STATE \hspace{1cm} Value at qubit: $\ell_{qb} = \log \frac{1-p}{p}+a$.
		\STATE \hspace{1cm} If $\ell_{qb} < 0$ do:
		\STATE \hspace{2cm} $\tilde \error = \tilde \error \cup {q}$
		\STATE Return $\tilde \error$
	\end{algorithmic}
\end{algorithm}

\medskip
Note that for better numerical stability we use logarithmic ratio of probabilities instead of probabilities :
\begin{align} \label{log_ratio_prob}
\ell(p)
= 
\log \frac {1-p} {p}
\end{align}

It is simpler to have no tuning and it was observed in \cite{breuckmann2020single_shot_4D_codes} that this variant of BP performs similarly to the standard PB. To achieve this goal we use the same approach as in \cite{breuckmann2020single_shot_4D_codes} and try every possible number of rounds starting with 1 round and incrementing this number by 1 at each step. Note that the computational overhead of this approach is reasonable since it is possible to re-use the $(R-1)$ first rounds of BP with $R$ rounds as the $(R-1)$ first rounds of BP with $(R+1)$ rounds. The termination condition depends on $\sigma_{res}^{(R)}$: the syndrome of the residual error. The residual error is the symmetric difference of the estimation of the error after $R$ rounds and the actual error:
\begin{align} \label{sigma_res}
\sigma_{res}^{(R)}
= 
\sigma (\tilde \error^{(R)} \, \Delta \, \error)
\end{align}
Tuning-free BP terminates after $R$ rounds if one of the two following conditions is met:
\begin{itemize}
\item $\sigma_{res}^{(R)} = \varnothing$
\item $|\sigma_{res}^{(R)}| \geq |\sigma_{res}^{(R-1)}|$
\end{itemize}

In the first case, we call a successful decoding if the residual error is not a logical error and a failure otherwise. In the second case, we call a failure. Termination of the tuning-free belief-propagation decoder is guaranteed after a number of rounds which is at most the weight of the initial syndrome $\sigma$.

\end{document}